%% file: report.tex
\documentclass[12pt]{article}
\pdfoutput=1
\usepackage[margin=0.8in]{geometry}
\usepackage{amsmath,amssymb,latexsym}
\usepackage{graphicx} 
\usepackage{booktabs}
\usepackage{hyperref}
\usepackage[numbers,sort&compress]{natbib}
\usepackage[toc,page]{appendix}

\usepackage{amsthm}
\newtheorem{theorem}{theorem}
\usepackage{authblk}

\input{abbr1}

\graphicspath{{figs/}}

\title{Neural Representations in Hybrid Recommender Systems:\\ Prediction versus Regularization}

\author[ \hspace{-1ex}]{Ramin Raziperchikolaei}
\author[ \hspace{-1ex}]{Tianyu Li}
\author[ \hspace{-1ex}]{Young-joo Chung}
\affil[ ]{Rakuten Institute of Technology, Rakuten, Inc.}
\affil[ ]{\textit {\{ramin.raziperchikola,tianyu.li,youngjoo.chung\}@rakuten.com}}
\date{} 

\hypersetup{
pdftitle={Neural Representations in Hybrid Recommender Systems: Prediction versus Regularization},
pdfauthor={Ramin Raziperchikolaei, Tianyu Li,  Young-joo Chung},
pdfkeywords={Hybrid Recommender Systems, Autoencoders, Neural Collaborative Filtering}
}

\begin{document}

\maketitle
\begin{abstract}
Autoencoder-based hybrid recommender systems have become popular recently because of their ability to learn user and item representations by reconstructing various information sources, including users' feedback on items (e.g., ratings) and side information of users and items (e.g., users' occupation and items' title). However, existing systems still use representations learned by matrix factorization (MF) to predict the rating, while using representations learned by neural networks as the regularizer. In this paper, we define the neural representation for prediction (NRP) framework and apply it to the autoencoder-based recommendation systems. We theoretically analyze how our objective function is related to the previous MF and autoencoder-based methods and explain what it means to use neural representations as the regularizer. We also apply the NRP framework to a direct neural network structure which predicts the ratings without reconstructing the user and item information. We conduct extensive experiments on two MovieLens datasets and two real-world e-commerce datasets. The results confirm that neural representations are better for prediction than regularization and show that the NRP framework, combined with the direct neural network structure, outperforms the state-of-the-art methods in the prediction task, with less training time and memory. 
\end{abstract}

\section{Introduction}
\label{s:intro}
The goal of recommender systems is to help users identify the items that best fit their personal tastes from a large set of items \cite{Ricci11}. To achieve this goal, recommender systems use different kinds of user and item information. One important source of information is the feedback of users on items, which could be implicit (e.g., click on a link, purchase, etc.) \cite{Rendle09,Dong19,Ricci11} or explicit (e.g., a rating between $1$ and $5$) \cite{Koren09,Ricci11}. On e-commerce platforms, predicting users' explicit feedback (e.g., ratings on reviews) is more desirable because it provides better insight about users' preferences. Therefore, we focus on predicting explicit feedback (i.e., ratings). 

The focus of this paper is on hybrid recommender systems, which use both feedback of the users on items and \emph{side information} to make prediction  \cite{Burkea07}. Side information of the users and items include the content of items (e.g., category, title, description, etc.) and profile of users (e.g., age, location, gender, etc.), respectively. Using the feedback and side information jointly helps the hybrid methods to overcome the limitation of other approaches that use either feedback or side information, and achieve state-of-the-art results \cite{Li18,Liu19}.

More specifically, assume we have a sparse rating matrix $\bR \in \mathbb{R}^{m \times n}$, where $m$ and $n$ are the number of users and items, respectively, $R_{jk}>0$ is the rating of the user $j$ on the item $k$, and $R_{jk}=0$ means the rating is unknown.  Assume the side information of all the users and items are represented by $\bX$ and $\bY$, respectively. \emph{The goal of hybrid methods is to predict the unknown ratings using the known ratings and the user and item side information.}

%The $i$th row of a matrix $\bH$ is shown by $\bH_{i,:}$ and the $j$th column is shown by $\bH_{:,j}$. The indicator function is denoted by $\ind(arg)$, which returns $1$ when $arg$ is true, and $0$ otherwise.
Deep neural networks have become popular in designing hybrid recommendation methods, mainly because of their ability in learning good representations. The most widely-used neural network structure in recommender systems has been (denoising) autoencoders \cite{Sedhain15,Li15,Wang15b,Strub16,Dong17,Zhang17b,Li18}. These methods define $\bg^{u}()$, $\bbf^{u}()$, $\bg^{i}()$, $\bbf^{i}()$ as the user's encoder, user's decoder, item's encoder, and item's decoder, respectively. The outputs of the encoders $\bg^{u}()$ and $\bg^{i}()$ are the learned neural representations of the users and items. They also consider $\bU \in \mathbb{R}^{m \times d}$ and $\bV \in \mathbb{R}^{n \times d}$ as the $d$-dimensional representations of the users and items, respectively. Their objective function can then be written as:
\begin{multline}
	\label{e:autoobj}
	\hspace{-2ex}\min_{\bU,\bV,\btheta}L(\bbf^{u}(\bg^{u}(\bR,\bX))) + L(\bbf^{i}(\bg^{i}(\bR,\bY))) + \lambda_1 \sum_{j,k} \ind(R_{jk}>0)  ||R_{jk} - \bU_{j,:} \bV_{k,:}^T ||^2 +\\  \lambda_2 ||\bU - \bg^{u}(\bR,\bX) ||^2 + \lambda_3 ||\bV - \bg^{i}(\bR,\bY) ||^2  + \text{reg.\ terms},
\end{multline}
where $\btheta=[\btheta_{f^u},\btheta_{g^u},\btheta_{f^i},\btheta_{g^i}]$ contains all the parameters of the two autoencoders and $\bU_{j,:}$ denotes the $i$th row of the matrix $\bU$.  The indicator function $\ind(arg)$ returns $1$ when $arg$ is true, and $0$ otherwise. The rating of the user $j$ on item $k$ is approximated by the dot product of the $\bU_{j,:}$ and $\bV_{k,:}$.

We divide the objective function of Eq.~\eqref{e:autoobj} into three parts:
\begin{enumerate}
\item The reconstruction losses in the first two terms try to reconstruct the ratings and the side information of the users and items. 
\item The MF in the third term decomposes the rating matrix into user and item representations, which will be used for the prediction later. 
\item The fourth and fifth terms try to keep the representations learned by MF in some distance from the neural representations. As we argue in Section.~\ref{s:prop}, these terms play the role of regularizer, which keep the representations from converging to the solution of MF. The hyper-parameters $\lambda_2$ and $\lambda_3$ determine how close the two representations should be from each other. 
\end{enumerate}

Note that the reconstruction loss in Eq.~\eqref{e:autoobj} is slightly different from one work to the others. In \cite{Li15}, each autoencoder reconstructs the repeated versions of the side information. In \cite{Dong17}, each autoencoder reconstructs both the ratings and the side information, where the side information has been added to each layer of the network. \citet{Li18} proposed to reconstruct multiple sources of side information of the users and items, in addition to the ratings.

Autoencoder-based methods optimize the objective of Eq.~\eqref{e:autoobj} by alternating over the following two steps: 1) fix the network parameters and optimize over $\bU$ and $\bV$ and 2) fix $\bU$ and $\bV$ and train the parameters of the two autoencoders.

This approach has three main issues. First, the motivation behind using neural representation for the regularization purpose is unclear. Also, it is difficult to decide how far/close the neural and MF representations should be from each other, i.e., it is difficult to set the hyper-parameters $\lambda_2$ and $\lambda_3$. Second, optimization is difficult and time-consuming because  1)  the autoencoders have many parameters, 2) the matrices $\bU$ and $\bV$ are huge, and 3) the matrices and parameters need to be optimized several times in alternation. Third, the dot product to predict ratings from representations $\bU$ and $\bV$ might not be sufficient to combine the two representations.

To solve the above issues, we introduce the \textbf{N}eural \textbf{R}epresentation for \textbf{P}rediction (\textbf{NRP}) framework that learns one set of user and item representations from the neural networks and uses them for the prediction directly, instead of using them as the regularizer. Here are the contributions of our paper:
\begin{itemize}
	\item  We propose the NRP framework and apply it to the autoencoder structure, analyze its objective function and optimal solution, and compare it with the previous approaches based on autoencoder and MF.
	\item We introduce a direct neural network structure integrated with our framework to obtain more expressive power and training efficiency. 
	\item We conduct experiments on two MovieLens datasets and two real-world e-commerce datasets and demonstrate that 1) neural representations perform better in prediction than regularization and 2) the direct neural network structure combined with the NRP framework outperforms previous methods, while having faster training and less memory usage.
\end{itemize} 

\section{Related work}
\label{s:related}
\paragraph{Factorization models.}
Matrix factorization (MF)  decomposes the rating matrix into two low-rank user and item matrices (representations), such that the dot product of the representations approximates the rating matrix\cite{Koren08,Takacs08,Koren09}. Since MF only uses the ratings, its performance degrades significantly when the rating matrix is highly sparse or when we have new users and items in the system (cold-start problem) \cite{Ricci11}. 

\paragraph{Autoencoders for hybrid recommender systems.}
These methods use the autoencoder structure to extract features from side information and combine such features with the feedback pattern to predict ratings.\cite{Li15,Wang15b,Strub16,Dong17,Zhang17b,Li18}.  Recent works \cite{Li15,Dong17,Li18} utilize all the side information and ratings of users/items by training two autoencoders. As explained in the introduction, these methods use MF's representations for the prediction. There are also cases where the MF and neural representations are mixed for prediction.  In CDL \cite{Wang15b} and AutoSVD \cite{Zhang17b}, the item's representation comes from the autoencoder and the user's representation comes from the MF. These methods use the side information of the items as the only source of information, and use dot product to combine the representations. In our method, both users' and items' representations are generated from the neural network, where the inputs are the user/item rating vectors and side information.

\paragraph{Deep learning to model implicit feedback in collaborative filtering.}
DMF \cite{Xue17}, NeuMF \cite{He17}, and DeepCF \cite{Dong19} are examples of the recent works which combine the deep networks and collaborative filtering (CF) to predict the implicit feedback. These methods do not use the autoencoder structure, as the goal of autoencoders is to extract representations from the side information. Our method differs from these methods in several ways. First, our approach is a general framework that can be applied to different network structures, such as autoencoders. Second, our method uses both side information and ratings to predict the ratings, while the input to the CF-based methods is ratings and/or ids. Third, our approach learns a single representation per user and item, while NeuMF and DeepCF learn two representations. Lastly, our method uses an MLP to map the representations to prediction, while DMF uses dot product.

\paragraph{Deep learning for content-based recommendation.}
The main idea is to extract representations from the side information such as users' profiles and items' descriptions. DSSM \cite{Huang13} maximizes/minimizes the similarity between the representations of the query and the clicked/not-clicked documents. MV-DNN \cite{Elkahky15} is an extension of the DSSM, where the document has multiple views. The problem definition here is different from ours, as these methods do not consider the previous ratings of the users on items. 

\paragraph{Deep Learning for high-order interactions.}
In web/app recommender systems, the input usually comes in the form of high-dimensional and sparse categorical features.  Factorization Machines \cite{Rendle10} proposed to model high-order interactions by learning an embedding vector per feature. This idea has been combined with the embedding layer of the neural networks in several recent works, such as Wide\&Deep\cite{Cheng16}, NFM \cite{He17a}, and DeepFM\cite{Guo17}. Similar to the content-based methods, the problem definition is different from us, as these methods do not consider the past ratings of the users on items.

\section{Our proposed method}
\label{s:prop}
Our main idea is to remove the MF terms and use the neural representations of the users and items for the prediction task. For concision, we call our method \textbf{NRP}, which stands for \textbf{N}eural \textbf{R}epresentation for \textbf{P}rediction. We first apply our framework to the autoencoder structure and show how it is related to the previous autoencoder-based methods. Then, we integrate this framework with a direct neural network structure.

\subsection{NRP with autoencoders}
\label{s:prop_auto}
Similar to the previous works, our model contains two autoencoders, one for the users and one for the items. The difference is that the encoders' outputs are the only user/item representations in our model. Here is our objective function:
\begin{align}
	\label{e:our_auto} 
  &\min_{\btheta} L(\bbf^{u}(\bg^{u}(\bR,\bX))) + L(\bbf^{i}(\bg^{i}(\bR,\bY))) + \\
  &\lambda_1 \sum_{j,k}  \ind(R_{jk}>0) ||R_{jk} - \bg^{u}(\bR_{j,:},\bX_{j,:})^T \bg^{i}(\bR_{:,k},\bY_{k,:}) ||^2. \nonumber
\end{align}

To have an apple-to-apple comparison between our approach and the previous ones, we use the same loss functions and the decoder and encoder structures as aSDAE \cite{Dong17} and DHA \cite{Li18}.

Our objective in Eq.~\eqref{e:our_auto}  gives three advantages over the previous works: 1) the hyper-parameters $\lambda_2$ and $\lambda_3$, from Eq.~\eqref{e:autoobj}, are removed, 2) the number of parameters decreased as we removed $\bU$ and $\bV$, which helps in faster training and saving memory, and 3) the network can be trained end-to-end, as there is no need to optimize over $\bU$ and $\bV$. In Section ~\ref{s:exp}. we will see that these advantages lead to better prediction performance.

We now analyze the objective function of Eq.~\eqref{e:our_auto}, compare it with the one in Eq.~\eqref{e:autoobj}, and explain why the neural representations act as a regularizer in previous works. First, we rewrite our objective in \eqref{e:our_auto} as follows:
\begin{align}
\label{e:auto_equ}
	\min_{\btheta,\bU,\bV} Q(\btheta, \bU,\bV) =& L(\bbf^{u}(\bg^{u}(\bR,\bX))) + L(\bbf^{i}(\bg^{i}(\bR,\bY))) + \lambda_1 \sum_{j,k}  \ind(R_{jk}>0) ||R_{jk} -  \bU_{j,:}^T \bV_{k,:}||^2 \nonumber \\
	&\text{s.t.} \quad \bU = \bg^{u}(\bR,\bX) \quad \text{and} \quad \bV =  \bg^{i}(\bR,\bY).
\end{align}
The objective functions in Equations \eqref{e:our_auto} and \eqref{e:auto_equ} are  equivalent, so we focus on comparing \eqref{e:auto_equ} with \eqref{e:autoobj}.

We consider two special cases of the objective function in Eq.~\eqref{e:autoobj}. First, consider the case where $\lambda_2=\lambda_3=0$, which makes the last two terms $0$. The first two terms can also be removed since they do not contain the user/item representations $\bU$ and $\bV$. So only the MF term remains.

The second case is when $\lambda_2=\lambda_3 \to \infty$. The following theorem shows that in this case the two objective functions  in \eqref{e:autoobj} and \eqref{e:auto_equ} will be equivalent (i.e. they have the same optimal solution).

\begin{theorem}
The objective function of the Eq.~\eqref{e:autoobj}, with $\lambda_2=\lambda_3 \to \infty$,  has the same optimal solution as the objective function of the Eq.~\eqref{e:auto_equ}.
\end{theorem}

\begin{proof}
Let us define a new vector $\bp=[\btheta,\bU,\bV]$, containing all the parameters of the problem. Note that $Q(\bp)$, defined in Eq.~\eqref{e:auto_equ}, contains the first three terms of Eq.~\eqref{e:autoobj} and contains all the terms of Eq.~\eqref{e:auto_equ}.
We assume that $\bar{\bp} = [\bar{\btheta},\bar{\bU},\bar{\bV}]$ is the optimal solution of \eqref{e:auto_equ} and $\bp^* = [\btheta^*,\bU^*,\bV^*]$ is the optimal solution of the Eq.~\eqref{e:autoobj} when $\lambda_2=\lambda_3 \to \infty$. We replace $\bp^*$ and $\bar{\bp}$ in Eq.~\eqref{e:autoobj} to get the following inequality:
\begin{multline}
	\label{e:helper1}
	\lim_{\lambda_2=\lambda_3 \to \infty} Q(\bp^*) + \lambda_2 ||\bU^* - \bg^{u}(\bR,\bX;\btheta_{g^u}^*) ||^2+  \lambda_3 ||\bV^* - \bg^{i}(\bR,\bY;\btheta_{g^i}^*) ||^2 \le \\ Q(\bar{\bp}) + \lambda_2 ||\bar{\bU} - \bg^{u}(\bR,\bX;\bar{\btheta}_{g^u}) ||^2+  \lambda_3 ||\bar{\bV} - \bg^{i}(\bR,\bY,\bar{\btheta}_{g^i})||^2  = Q(\bar{\bp}).
 \end{multline}
The right side of the above inequality is $Q(\bar{\bp})$ since $\bar{\bp}$ is the optimal solution of Eq.~\eqref{e:auto_equ} and satisfies the constraints. By rearranging the last equation in \eqref{e:helper1}, we obtain
\begin{multline}
 ||\bU^* - \bg^{u}(\bR,\bX;\btheta_{g^u}^*) ||^2+ ||\bV^* - \bg^{i}(\bR,\bY,\btheta_{g^i}^*) ||^2 \le \lim_{\lambda_2 \to \infty} \frac{ Q(\bar{\bp}) - Q(\bp^*)   }{\lambda_2} = 0.
\end{multline}
To satisfy the inequality, i.e., making the sum of the norms $0$,  both norms have to be $0$. This means that $\bp^*$ is a feasible vector, i.e. $\bU^* = \bg^{u}(\bR,\bX;\btheta_{g^u}^*)$ and $\bV^* = \bg^{i}(\bR,\bY,\btheta_{g^i}^*)$. By considering feasibility of $\bp^*$ in Eq.~\eqref{e:helper1} (i.e., replacing $\bU^*$ by $ \bg^{u}(\bR,\bX;\btheta_{g^u}^*)$ and $\bV^*$ by  $\bg^{i}(\bR,\bY,\btheta_{g^i}^*)$ in Eq.~\eqref{e:helper1}), we get $Q(\bp^*) \le Q(\bar{\bp})$. Note that we assumed that $\bar{\bp}$ is the optimal solution of Eq.~\eqref{e:auto_equ}, so for two feasible points $\bar{\bp}$ and $\bp^*$ we have $Q(\bp^*) \ge Q(\bar{\bp})$. So the conclusion is that $Q(\bp^*) = Q(\bar{\bp})$. In other words, for $\lambda_2=\lambda_3 \to \infty$, the objective functions of the Eq.~\eqref{e:autoobj} and Eq.~\eqref{e:auto_equ} become equivalent.
\end{proof}

\begin{figure}
	\centering
	\begin{tabular}{c}
	 	 \includegraphics[width=.6\linewidth]{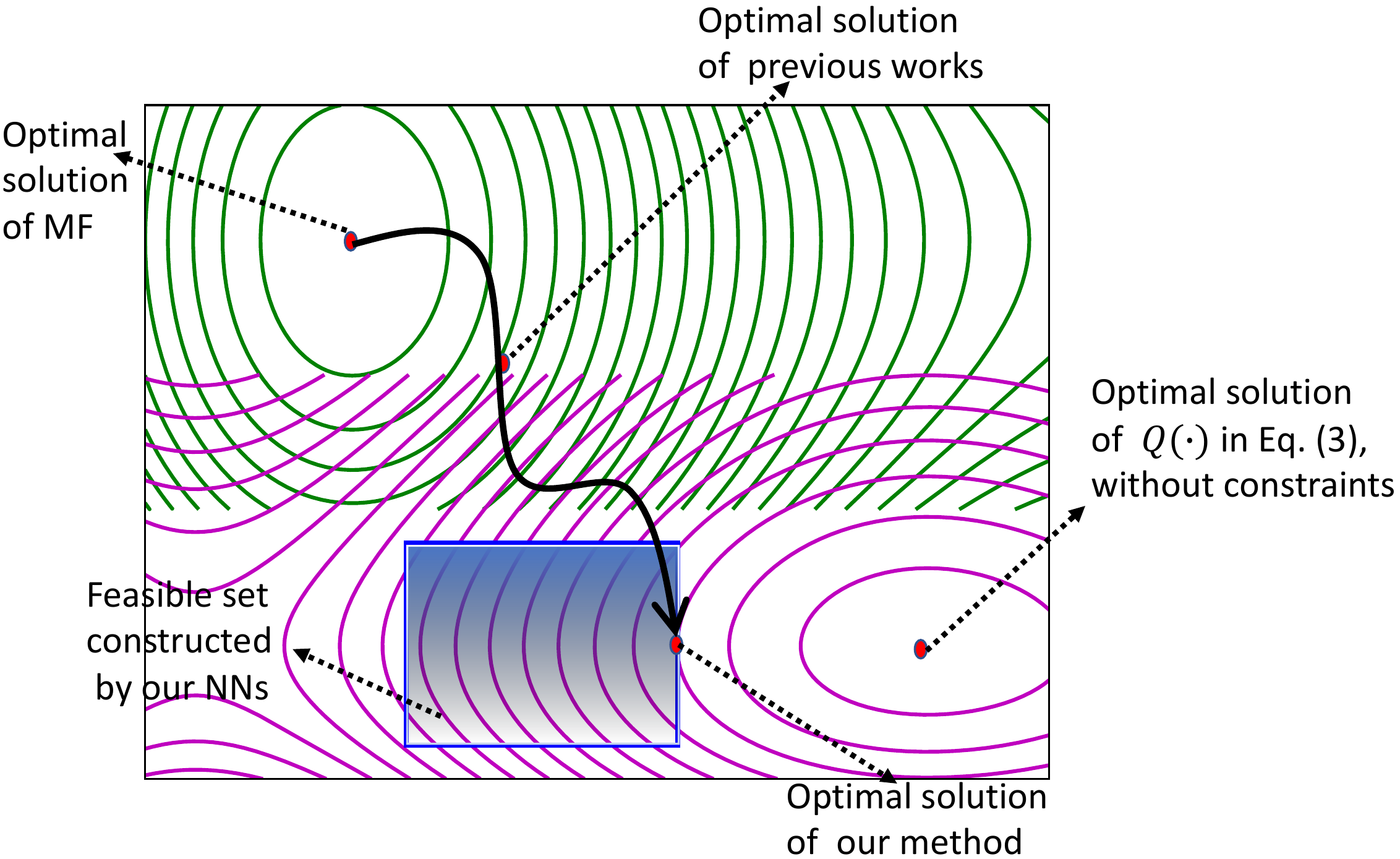}\\
	\end{tabular}
	\vspace{-1.8ex}
	\caption{Visualization of the optimal solutions of different methods. The figure shows the contours over the users and items representations ($\bU$ and $\bV$). \emph{Green contours:} contours of the MF, which is achieved by setting $\lambda_2=\lambda_2=0$ in Eq.~\eqref{e:autoobj}. \emph{Magenta contours:} contours of $Q(\cdot)$, the main term of our objective function in Eq.~\eqref{e:auto_equ}.  There is a path between the optimal solution of MF and ours. Previous approaches find a solution somewhere in the middle of the path. Our solution is always inside the feasible set created by the encoders.}\vspace{-1.5ex}
	\label{f:contour}
\end{figure}
%changed

Fig.~\ref{f:contour} shows a simple visualization of the objective functions and their optimal solutions. In this figure, the green contours correspond to the MF (by setting $\lambda_2=\lambda_3=0$ in Eq.~\eqref{e:autoobj}) and  the magenta contours correspond to $Q(\cdot)$, the main term of our objective function in Eq.~\eqref{e:auto_equ}. The feasible set, which satisfies our constraints in Eq.~\eqref{e:auto_equ}, has been shown by a blue rectangle. This feasible set contains the low-dimensional representations that can be created by the user and item encoders. The optimal solution of our objective function in Eq.~\eqref{e:auto_equ} lies where the contour line of $Q(\cdot)$ with the smallest value intersects the feasible region. 

By setting $\lambda_2=\lambda_3=0$ and increasing it to $\lambda_2=\lambda_3 \to \infty$, a path of solutions will be created, between the solution of the MF and our NRP autoencoder. The previous autoencoder methods use a fixed $\lambda_2>0$ and $\lambda_3>0$, so their optimal solution lies somewhere on the path. The smaller (larger) these hyper-parameters, the closer (farther away) the solution of the autoencoder-based methods will be to the MF's solution. \emph{We believe the neural representations act as the regularizer in previous works since they are only used to keep $\bU$ and $\bV$ away from the MF's optimal solution.}

A question arises here: can we optimize the objective of Eq.~\eqref{e:autoobj} with $\lambda_2=\lambda_3 \to \infty$ and expect to get the same result as our objective function in Eq.~\eqref{e:our_auto}? In practice, this will not happen for several reasons. First, as we set $\lambda_2=\lambda_3 \to \infty$, the Hessian of the objective function in Eq.~\eqref{e:autoobj}  becomes ill-conditioned, and the contours will have a "banana" shape rather than "elliptical" shape. The reason is that some eigenvalues of the Hessian will be in the order of $\infty$, while the other ones will be small. More information can be found in Chapter 17.1 of \cite{Nocedal06}. As a result, first-order methods iterate zigzag and slowly approach toward the optimal solution.

Second, note that previous autoencoder-based methods use alternating optimization over the parameters. Consider the step of optimizing over $\bU$ and $\bV$ for a fixed $\btheta$:
\begin{equation*}
		\hspace{-2ex}\min_{\bU,\bV} \lim_{\lambda_2=\lambda_3=\to \infty} \lambda_1 \sum_{j,k} \ind(R_{jk}>0)  ||R_{jk} - \bU_{j,:}^T \bV_{k,:} ||^2 +  \lambda_2 ||\bU - \bg^{u}(\bR,\bX) ||^2 + \lambda_3 ||\bV - \bg^{i}(\bR,\bY) ||^2 .
\end{equation*}
Note that even a small mismatch between the neural's and MF's representations, the error in the second and the third terms, makes the objective error $\infty$. So the solution to the above optimization problem will be $\bU=\bg^{u}(\bR,\bX)$ and $\bV=\bg^{i}(\bR,\bY)$. In other words, the first term, which is responsible for rating prediction, will always be ignored and the performance degrades significantly.

Finally, note that we have proved the equivalency of Eq.~\eqref{e:autoobj} and \eqref{e:our_auto} in terms of their optimal solutions. In practice, the objective functions are highly non-convex and we can expect the methods to end up close to a local solution.  As we optimize the neural network using SGD, our method starts from some initial point, traverses a path, and stops as soon as overfitting happens.  \emph{From start to end, the solution of our method remains inside the feasible set, while the final solution of the autoencoder-based methods can be anywhere in the space.}

\paragraph{Size of the feasible set.} The constraints of the  Eq.~\ref{e:auto_equ} determines the feasible set. Assuming $d$-dimensional representations, the size of the feasible set is determined by the number of $d$-dimensional vectors that can be generated by the user and item encoders, which depends on the complexity of the model, i.e., increasing the complexity of the encoders leads to a larger feasible set. In our visualization in Fig.~\ref{f:contour}, as the feasible set gets larger, the optimal solution of our method gets closer to the optimal solution of $Q(\cdot)$. As discussed before, the optimal solution of the previous autoencoder-based methods always stays somewhere in the middle of the path.

\begin{figure}
	\centering
 	\includegraphics[width=.6\linewidth]{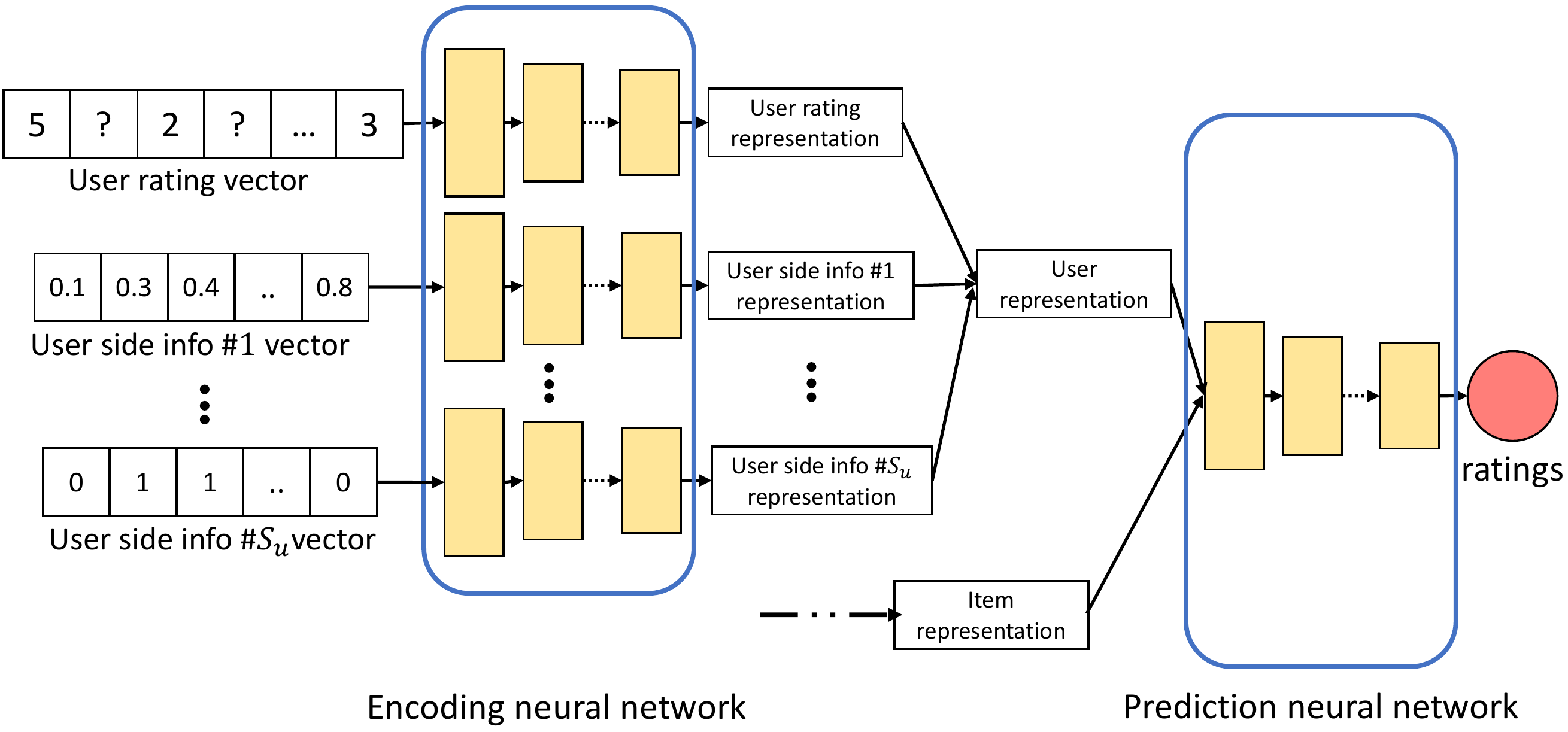}
 	\vspace{-1ex}
	\caption{The diagram of our direct neural network structure. Encoding neural network maps the multiple sources of the input to a low-dimensional user representation. Item representation is created by the same structure. Prediction neural network takes the joint (user,item) representation and predicts the rating.}\vspace{-1.5ex}
	\label{f:diag}
\end{figure}
\subsection{NRP with a direct structure}
\label{s:prop_direct}
We define the direct structure as a neural network structure without the decoders, which predict the ratings without reconstructing the user/item ratings and side information. Our main goal of designing the direct structure and combining it with NRP framework is to use it as a baseline. The comparison between the direct structure and the autoencoder-based methods let us know the effectiveness of the reconstruction based methods in hybrid recommendation systems. 

Our direct structure is achieved by making two modifications to our autoencoder structure. First, we remove the decoders from the structure, which leads to saving around 50\% of memory and faster optimization. Second, we use a set of fully connected layers to predict the final rating, instead of the dot product. This makes our model more expressive. Note that these fully connected layers add a small amount of memory to our model since they connect the low-dimensional user and item representations to the final rating. On the other hand, the decoders are usually huge, since they reconstruct the input ratings with thousands to millions of dimensions.

Fig.~\ref{f:diag} shows the overall architecture of our direct structure. We define three cooperating sub-networks: user encoding network, item encoding network, and prediction network. The user and item encoding networks generate the user and item representations, respectively, which are combined to get the joint (user,item) representation. Finally, the prediction network takes the joint representation and predicts the rating.  

\paragraph{Learning the user and item representations.} We focus on the user representation, as the item representation can be learned in the same way. We have access to multiple sources of information for each user, including ratings of the user on a subset of items, and $S_u$ sources of side information. The user encoding network maps each source to a low-dimensional space and then combines them to get the final user representation. 

One source of information is the ratings of the users on items. For the $j$th user, $\bR_{j,:}$ contains the ratings of this user on all items. We give this as the input to a set of fully connected layers and get the rating representation $\bg^{u}_{\text{rating}}(j) \in \bbR^{1\times d_{\text{rating}}}$ for the $j$th user.

We also learn one representation for each source of side information. Each source will be given as the input to a neural network, which maps it to a low-dimensional representation. The representation for the $s_u$th source is shown by $\bg^{u}_{s_u}(j) \in \bbR^{1\times d_{s_u}}$. The structure of the network for different sources could be different. For example, it could be a convolutional network for the images, a long-short term memory (LSTM) network for the text descriptions and titles, and fully connected layers for the other general type of inputs.

To get the final representation for the $j$th user, we concatenate the representations of all the sources. We show this by $\bz_j = [\bg^{u}_{\text{rating}}(j),\{\bg^{u}_{s_u}(j)\}_{s_u=1}^{S_u}] \in \bbR^{1\times d_{u}}$, where $[ \quad]$ is used for concatenation of the vectors.

We can get the $k$th item representation $\bz_k$ with minor changes to the user representation explained above. For the rating representation, the input should be a column of the rating matrix $\bR$, instead of the row. We concatenate the representations of the different sources of information to get the final item representation $\bz_k$.

\paragraph{Learning the rating.}
We first concatenate the user and item representations to get the joint representation $\bz_{jk}$. This joint representation is the input to another set of fully connected networks that try to predict the rating of the users on the items. We show this last neural network by $h()$. We define the objective function as the mean squared error between the output of our neural network and the true rating: 
\begin{align}
	\label{e:learn}
	&\frac{1}{mn}\min_{\btheta} \sum_{j=1}^{m}\sum_{k=1}^n \ind(R_{jk}>0)  ||R_{jk} - h(\bz_{jk}))||^2 + \lambda_1 ||\btheta||^2 \nonumber  \\
	 & \text{s.t.} \quad
       \bz_{jk} = [\bz_j,\bz_k], \  \bz_{j} = [\bg_{\text{rating}}^u(j),\{\bg_{s_u}^u(j)\}_{s_u=1}^{S_u}],\ 
		\bz_{k} = [\bg_{\text{rating}}^i(k),\{\bg_{s_i}^i(k)\}_{s_i=1}^{S_i}]
\end{align}
We jointly optimize all parameters of the neural network using stochastic gradient descent. Note that our method works as long as we have at least one source of information for users and one source of information for items.

\paragraph{Direct structure in the literature.} We noticed that different variants of the direct structure have been used in other areas of the recommendation systems, such as implicit feedback prediction from user-item interactions \citep{He17,Dong19} and modeling high-order user-item interactions \citep{Cheng16,Guo17}.  To the best of our knowledge, the direct structure has been used once in hybrid recommender systems by \citet{Shi18} (denoted by ACCM) to address the cold-start issue. There are three main differences between our direct structure and ACCM: 1) our approach uses interaction vector as the input, while ACCM uses user/item ID as the input, 2) our structure uses MLPs to map user-item interaction to the rating, while ACCM uses the dot product, and 3) our approach uses concatenation to combine user and item representation, while ACCM uses the weighted sum. We will show in our experiments that ACCM achieves comparable results to the autoencoder-baed methods, while our direct structure outperforms them. 

\paragraph{ID of the users and items as another source of information}  It is easy to adjust our model and objective function to include other sources, such as IDs. In our experiments, we found that adding IDs as another source does not improve the performance.  Since IDs might be useful in some other datasets, we briefly explain how to add this source to our model.

Let's focus on the user ID first. To learn the low-dimensional representation for the user IDs, we use an embedding layer. The embedding layer works as a look-up table; it takes the user ID and returns one row, which contains the embedding vector of the user. We  denote the embedding vector of the $j$th user by $\bg^{u}_{\text{ID}}(j)$ and add it to the $\bz_j$ in Eq.~\ref{e:learn}. The embedding vector of the $k$th item, $\bg^{i}_{\text{ID}}(k)$,  can be added to the $\bz_k$ in the same way. The rest of our objective function does not change.

Let's see how it works mathematically for the user embedding. We encode each user ID to a one-hot vector of size $m$, where $m$ is the total number of users. For the $j$th user, $\bI_{j,:}$ shows this encoding, where $\bI$ is an identity matrix of size $m\times m$. This one-hot vector is connected to a layer with the weight matrix $\bE^{(u)} \in \bbR^{m\times d_e}$. The output of the embedding layer is achieved by matrix multiplication $\bg^{u}_{\text{ID}}(j)=\bI_{j,:}\bE^{(u)} \in \bbR^{1\times d_e}$, which returns the $j$th row of $\bE^{(u)}$  containing the $d_e$ dimensional embedding of the $j$th user.

\section{Experiments}
\label{s:exp}
For each dataset, we randomly select $80\%$ of the ratings as the training set, $10\%$ as the validation set, and $10\%$ as the test set. We repeat the process three times to create three training/validation/test sets and report the mean and standard deviation of each method on these three sets. Unless otherwise stated, we use bag of words (BoW) to represent the item and user side information.

\paragraph{Datasets.}
We use four datasets in our experiments. Table~\ref{t:ds} lists the number of users and items, and the amount of sparsity of each dataset.
\begin{enumerate}
\item ml100k \cite{Harper15}. The dataset contains $100\,000$ ratings ($1$ to $5$) from around $1\,000$ users on $1\,600$ movies. The user side information contains age, gender, occupation, and zip code; the size of the feature vector is $879$. The item side information contains the movie title and the genre; the size of the feature vector is $2\,479$.
\item ml1m \cite{Harper15}. It contains $1$ million ratings of around $6\,000$ users on $4\,000$ movies. The content of the user and item side information are the same as ml100k. The dimensions of the user and item side information are $92$ and $4\,606$, respectively.
\item Amazon review data: Grocery and Gourmet Food \cite{He16a}. The dataset contains $508\,800$ Amazon's Grocery and Gourmet Food product review from $86\,400$ users on $108\,500$ items. There is no user side information. The item side information contains the title and the category, and is represented by a $36\,258$ dimensional vector.
\item Ichiba. It contains $1.5$ million Rakuten Ichiba\footnote{\url{https://rit.rakuten.co.jp/data_release/}}'s product review from $324\,000$ users on $294\,000$ items. The user side information contains age and gender represented by $121$ dimensional vector. The item side information contains the category with $455$ dimensions.
\end{enumerate}

\begin{table}[!t]
\caption{Summary of the four datasets. }
\vspace{-2ex}
\label{t:ds}
\begin{center}
\begin{tabular}{cccc} 
\toprule
     Dataset & \# of users & \# of items & sparsity\\
\midrule
     ml100k & $1\,000$ & $1\,600$ & $94\%$\\
\midrule
	ml1m & $6\,000$ & $4\,000$ & $96\%$  \\
\midrule
	Amazon R. & $86\,400$ & $108\,500$ & $99.994\%$  \\
\midrule
	Ichiba & $324\,000$ & $294\,000$ & $99.84\%$  \\	
\bottomrule
\end{tabular}
\end{center}
\vspace{-1.5ex}
\end{table}

\paragraph{Evaluation metrics.} We use the root mean square error (RMSE) and precision to evaluate the prediction performance. Let us assume set $T$ contains all the ratings in the test set, where $R_{jk} \in T$ is the actual rating of the user $j$ on item $k$. We define $\hat{R}_{jk}$ as the predicted rating, generated by a recommender system. Then, RMSE is defined as follows:
\begin{equation}
\text{RMSE} =\textstyle \sqrt{\frac{1}{|T|} \sum_{R_{jk}\in T} (R_{jk} - \hat{R}_{jk})^2}.
\end{equation} 

To report precision, we need to define the set of retrieved items and the set of relevant items per user. We define set $S_j$ as the set of items rated by the user $j$. To define the relevant items (groundtruth) for user $j$, we sort the items in $S_j$ based on their rating and pick the top $p\%$. At the test time, for user $j$, we predict the ratings of the items in  $S_j$ and consider the top $p\%$ of the items (with the highest predicted ratings) as the retrieved set. We define the precision for user $j$ as:
\begin{equation}
\text{precision} = \frac{|\text{relevant items}| \cap |\text{retrieved items}|}{|\text{retrieved items}|}
\end{equation}
We report the average precision of all the users in test set.

\begin{table}
\caption{Our proposed NRP framework, trained with the autoencoder and direct structures, versus MF, autoencoder-based methods, and ACCM (the direct structure of  \citet{Shi18}) on ml100k and ml1m datasets. We report mean and standard deviation of the methods. The first/second number in the fourth column is the number of parameters involved in learning the neural network's/MF's representations. Time refers to the training time per epoch. Our framework combined with the direct structure achieves better prediction results, faster training, and less memory usage compared to the autoencoder-based methods.} 
\label{t:auto}
\label{t:ml100aut}
\begin{center}
\begin{tabular}{@{}c@{\hspace{1ex}}c@{\hspace{1ex}}c@{\hspace{1ex}}c@{\hspace{1ex}}c@{\hspace{4ex}}}  
\toprule
& \multicolumn{4}{c}{\hspace{-3ex} \dotfill ml100k \dotfill}  \\
     method & RMSE & precision & \#  params. & time \\
\midrule
MF & $0.940 \pm 0.003$ & $68.4\% \pm 0.5$ & ($0$, $0.26$M)& $15$s \\
\midrule
DHA    &  $0.939 \pm 0.002$ & $68.2\% \pm 0.6$ & ($8.6$M, $0.26$M) & $85$s \\
\midrule
\textbf{NRP\textsubscript{DHA}} & $0.926 \pm 0.004$  & $68.4\% \pm 0.6$ &  ($8.6$M, $0$) & $68$s \\
\midrule
aSDAE  & $0.946 \pm 0.005$ & $68.0\%  \pm 1.1$&  ($13$M, $0.26$M) & $98$s \\
\midrule
\textbf{NRP\textsubscript{aSDAE}} & $0.910 \pm 0.008$ & $69.0 \% \pm 0.2$ & ($13$M, $0$) & $70$s \\
\midrule
ACCM  & $0.928 \pm 0.004$ & $68.2\% \pm 0.1$&  ($3.1$M, $0$) & $37$s \\
\midrule
ACCM\textsubscript{MLP}  & $0.925 \pm 0.005$ & $67.7\% \pm 0.3$&  ($3.3$M, $0$) & $37$s\\
\midrule
\textbf{NRP\textsubscript{direct}} & $\mathbf{0.899 \pm 0.006}$ & $\mathbf{70.2 \% \pm 0.4}$ &  ($4.7$M, $0$) & $42$s  \\
\bottomrule
\end{tabular}
\end{center}

\begin{center}
\begin{tabular}{@{}c@{\hspace{1ex}}c@{\hspace{1ex}}c@{\hspace{1ex}}c@{\hspace{1ex}}c@{\hspace{4ex}}}  
\toprule
& \multicolumn{4}{c}{\dotfill ml1m \dotfill} \\
     method & RMSE & precision & \#  params. & time \\
\midrule
MF & $0.892 \pm 0.004$ & $68.2\% \pm 0.3$ & ($0$, $1$M)& $45$s\\
\midrule
DHA  & $0.865 \pm 0.001$ & $69.3\% \pm 0.2$ & ($44$M, $1$M)& $1\,097$s\\
\midrule
\textbf{NRP\textsubscript{DHA}} & $0.855 \pm 0.002$ & $69.6\% \pm 0.2$ & ($44$M, $0$)& $1\,027$s \\
\midrule
aSDAE  & $0.879 \pm 0.005$ & $69.0\% \pm 0.1$ & ($66$M, $1$M)& $1\,155$s \\
\midrule
\textbf{NRP\textsubscript{aSDAE}} & $0.877 \pm 0.008$ & $68.5\% \pm 0.4$ & ($66$M, $0$)& $1\,055$s\\
\midrule
ACCM  & $0.856 \pm 0.002$ & $69.5\% \pm 0.3$ & ($11.5$M, $0$)& $450$s \\
\midrule
ACCM\textsubscript{MLP} & $0.865 \pm 0.002$ & $68.9\% \pm 0.2$ & ($11.8$M, $0$)& $470$s \\
\midrule
\textbf{NRP\textsubscript{direct}}  & $\mathbf{0.851 \pm 0.001}$ & $\mathbf{70.0 \% \pm 0.1 }$ &  ($22$M, $0$) & $640$s  \\
\bottomrule
\end{tabular}
\end{center}

\end{table}

\paragraph{Experimental setting.}
We implement our method using Keras with TensorFlow 1.12.0 backend. We ran all the experiments on a $12$GB GPU. For each method, we tried a set of activation functions (relu, selu, and tanh), a range of learning rates and regularization parameters from $10^{-1}$ to $10^{-5}$, a set of optimizers (Adam, SGD, and RMSprop), and picked the one that works best. For a fair comparison, all autoencoder methods have the same structure (\# of layers, neurons, etc.). Supplementary material at the end of this script contains the details of the experimental setting.

\paragraph{Neural representations are better in prediction than regularization.} In Table~\ref{t:auto}, we compare our proposed method with matrix factorization (MF), the autoencoder-based methods, DHA  \cite{Li18} and aSDAE \cite{Dong17}, and the attention-based direct structure, ACCM \cite{Shi18}, on MovieLens datasets. We also replaced the dot product of ACCM with MLP,  denoted by ACCM\textsubscript{MLP} in Table~\ref{t:auto}. Our framework applied to the same encoder-decoder structure as DHA and aSDAE are called  NRP\textsubscript{DHA} and NRP\textsubscript{aSDAE}, respectively. Our framework applied to the direct structure is called NRP\textsubscript{direct}. 

MF has the worst performance (the largest RMSE), which means that MF formulation with the L2 norm of weights as the regularization is not enough for the rating prediction. By setting the hyper-parameters carefully, both  DHA and aSDAE outperform SVD, which suggests that the neural network's representations are better regularizers than the L2 norm of the weights. Our methods NRP\textsubscript{DHA} and NRP\textsubscript{aSDAE} outperform the original DHA and aSDAE, respectively, in terms of RMSE and precision. This improvement comes from removing the MF terms, relying on neural representations, and staying inside the feasible set of neural network's output. ACCM, the direct structure of \citet{Shi18}, outperforms MF, DHA, and aSDAE, but shows similar performance to the NRP\textsubscript{DHA}. So by using the direct structure of ACCM we cannot conclude that the decoders are unnecessary to get the best results. Finally, NRP\textsubscript{direct} has the best performance and outperforms all autoencoder-based methods. This shows that removing the decoders, making the neural network free of reconstructing the inputs, and replacing the dot product with MLPs lead to learning a better model.

Note that NRP\textsubscript{direct} outperforms ACCM because of  1) using interaction vectors as the input instead of IDs and 2) using MLPs, instead of the dot product, to map the joint representation to the rating. To show this, we have replaced the dot product of ACCM with MLPs, denoted by  ACCM\textsubscript{MLP} in Table~\ref{t:auto}. We can see that the result of ACCM\textsubscript{MLP} is slightly better than ACCM in ml100k and slightly worse in ml1m. NRP\textsubscript{direct} outperforms ACCM\textsubscript{MLP}  in both datasets.

The last two columns of Table~\ref{t:auto} compare the number of learnable parameters and training time per epoch of each method.  DHA and aSDAE have the largest memory usage and training time, as they optimize the two representations alternatively. NRP\textsubscript{DHA} and NRP\textsubscript{aSDAE} achieve better training time and memory than DHA and aSDAE, respectively. Among the neural network-based methods, the direct structures, ACCM and NRP\textsubscript{direct},  have the fastest training and lowest memory usage because of their simple structure. Note that NRP\textsubscript{direct} has a larger number of parameters than ACCM. This is because NRP\textsubscript{direct} uses MLPs to map the interaction vectors to the low-dimensional representations, while ACCM uses embedding layers. 

\begin{table}[!t]
\caption{We report the precision by creating the relevant and retrieved sets using top $10\%$ and $25\%$ of the items. We put "OM" in the tables whenever we get an out-of-memory error.}
\vspace{-2ex}
\label{t:comp3}
\begin{center}
\begin{tabular}[c]{ccccc} 
\toprule
& \multicolumn{2}{c}{ml1m} & \multicolumn{2}{c}{ Amazon review} \\
 method & top $10\%$ & top $25\%$ & top $10\%$ & top $25\%$ \\
\midrule
MF & $55.6\% \pm 0.16$ & $68.05\% \pm 0.45$ & $64.9\% \pm 0.04$ & $71.5\% \pm 0.67$ \\
\midrule
Autorec & $57.6\% \pm 0.26$ & $69.5\% \pm 0.43$  & $62.6\% \pm 0.98$ & $69.8\% \pm 0.62$ \\
\midrule
NeuMF & $56.8\% \pm 0.12$ & $68.9\% \pm 0.48$  & $66.8\%  \pm 0.30 $ & $72.6\% \pm 0.09$ \\
\midrule
DSSM & $54.7\% \pm 0.35$ & $67.2\% \pm 0.30$ & NA & NA\\
\midrule
DHA & $57.4\% \pm 0.78$ & $69.3\% \pm 0.23$ & OM & OM\\
\midrule
NRP\textsubscript{DHA}  &  $57.3\% \pm 0.17$& $69.5\% \pm 0.32$  &  $66.6\% \pm 0.60 $& $72.9\% \pm 0.63$\\
\midrule
aSDAE & $56.4\% \pm 0.39$ & $68.7\% \pm 0.42$ & OM & OM\\
\midrule
NRP\textsubscript{aSDAE}  & $57.1\% \pm 0.3$ & $69.0\% \pm 0.41$&  $64.5\% \pm 0.43$& $71.2\% \pm 0.68$ \\
\midrule
HIRE   & $57.4\% \pm 0.08$ & $69.4\% \pm 0.55$&  OM & OM  \\
\midrule
NRP\textsubscript{direct} & $\mathbf{58.1\% \pm 0.16}$  & $\mathbf{69.9\% \pm 0.42}$ &   $\mathbf{67.3\%  \pm 0.38}$& $\mathbf{73.1\% \pm 0.24}$ \\
\bottomrule
\end{tabular}
\end{center}
\vspace{-1.7ex}
\end{table}

\begin{table}[!t]
\caption{RMSE of our NRP framework compared with the hybrid and collaborative filtering methods. Our approach outperforms the rest of the methods. }
\vspace{-2ex}
\label{t:comp1}
\begin{center}
\begin{tabular}[c]{ccccc} 
\toprule
 method & ml100k & ml1m  & Amazon & Ichiba\\
\midrule
MF & $0.940 \pm 0.003$ & $0.892 \pm 0.0004$ & $1.153 \pm 0.003$ & $1.00 \pm 0.104$ \\
\midrule
Autorec & $0.921 \pm 0.002$ &   $0.889 \pm 0.0003$ & $2.19 \pm 0.01$ &  $2.47 \pm 0.059$  \\
\midrule
NeuMF  & $0.948 \pm 0.005$ &  $0.886 \pm 0.001$ & $1.140 \pm 0.004$ &  $0.900 \pm 0.004$ \\
\midrule
DSSM & $0.934 \pm 0.002$ & $0.941 \pm 0.0004$ & NA & $0.913 \pm 0.003$ \\
\midrule
DHA & $0.939 \pm 0.002$ & $0.865 \pm 0.001$ & OM & OM \\
\midrule
NRP\textsubscript{DHA}   &  $0.926 \pm 0.004$  &   $0.855 \pm 0.002$ &  $1.135 \pm 0.002$&   OM \\
\midrule
aSDAE & $0.946 \pm 0.005$ & $0.879 \pm 0.005$ & OM & OM \\
\midrule
NRP\textsubscript{aSDAE}   &  $0.910 \pm 0.008$ & $0.877 \pm 0.008$ & $1.24 \pm 0.004$  &  OM  \\
\midrule
HIRE   &  $0.930 \pm 0.006$ & $0.861 \pm 0.004$ &  OM & OM  \\
\midrule
NRP\textsubscript{direct} & $\mathbf{0.897 \pm 0.003}$   &  $\mathbf{0.851 \pm 0.001}$ & $\mathbf{1.135 \pm 0.002}$  &  $\mathbf{0.889 \pm 0.002}$ \\
\bottomrule
\end{tabular}
\end{center}
\vspace{-2ex}
\end{table}

\paragraph{Comparison with other hybrid and collaborative filtering methods.} We compare our method with several baselines and recent works in Tables~\ref{t:comp3} and \ref{t:comp1}. MF \cite{Koren09} is a baseline collaborative filtering method that uses dot product of the user and item representations to predict the rating. Autorec \cite{Sedhain15} is another collaborative filtering method which uses autoencoders to reconstruct the ratings. NeuMF \cite{He17} combines deep and shallow networks and uses the user/item ID as the input to predict the ratings. NeuMF was originally proposed for the prediction of implicit feedback, but it can easily be modified to work with explicit feedback. DSSM \cite{Huang13} is a content-based recommender method, which uses deep neural networks to learn the representations. We modify DSSM to make it applicable to the explicit feedback prediction by connecting the user/item representations to a MLP with a mean squared error loss. DSSM is not applicable in the Amazon dataset, as there is no user side information. HIRE \cite{Liu19} is a hybrid method that considers the hierarchical user and item side information. DHA \cite{Li18} and aSDAE \cite{Dong17} are autoencoder-based methods, which use neural representations as the regularizer.

In Table~\ref{t:comp3} we report precision and in Table~\ref{t:comp1} we report RMSE of methods on four datasets. We can see that our method achieves the best results on different datasets.

\begin{table}[!t]
\caption{ RMSE of our method NRP\textsubscript{direct}, which is trained with and without user and item side information, on ml100k and Ichiba datasets. Side information helps in achieving better prediction performance.}
\vspace{-2ex}
\label{t:prof}
\begin{center}
\begin{tabular}[c]{ccccccccc} 
\toprule
& \multicolumn{2}{c}{ml100k} & \multicolumn{2}{c}{ Ichiba} \\
  &  precision   & RMSE  & precision   & RMSE  \\
\midrule
no side info  &  $0.901 \pm 0.006$ & $69.9\% \pm 0.43$   &  $0.895 \pm 0.003$ & $78.9\% \pm 0.002$  \\
\midrule
side info & $\mathbf{0.897 \pm 0.003}$ & $\mathbf{70.2\% \pm 0.42}$ & $\mathbf{0.889 \pm 0.002}$  &   $\mathbf{80.1\% \pm 0.004}$  \\
\bottomrule
\\[.5ex]
\end{tabular}
\end{center}
\vspace{-5ex}
\end{table}

\paragraph{Importance of side information.}
 We train our model with and without the user/item side information to verify that the side information can improve the prediction results. Table~\ref{t:prof} lists the RMSE and precision results on different datasets. We can see that the side information helps to achieve better results.
 
 \section{Conclusion}
The current autoencoder-based hybrid recommender systems learn two types of representations. One comes from the matrix factorization and is used for prediction. The other comes from neural networks and is used for regularization. In this paper, we proposed a new framework that uses the neural networks' representation directly for the prediction task. We showed that by applying our approach to the same autoencoder structure as previous works, we can achieve faster training and better performance. We also proposed a simpler network structure by removing the decoders and replacing dot product with MLP in autoencoders. Our approach combined with the new proposed framework outperformed the previous works. It also had a fast training and small memory usage compared to the autoencoder-based methods.

\small{
\bibliographystyle{plainnat}

}

\newpage
\section*{Supplementary material: experimental settings}
We implement our method using Keras with TensorFlow 1.12.0 backend. We ran all the experiments on a $12$GB GPU. For each method, we tried a set of activation functions (relu, selu, and tanh), a range of learning rates and regularization parameters from $10^{-1}$ to $10^{-5}$, a set of optimizers (Adam, SGD, and RMSprop), and picked the one that works best. For a fair comparison, all autoencoder methods have the same structure (\# of layers, neurons, etc.). Table~\ref{t:imp} lists these details for all the methods.

In the following, we give more information about the structure of the neural network and other hyper-parameters of each method. The notation $[a,b,c]$ denotes a network with three fully connected layers, where $a,b,$ and $c$ are the number of neurons in each layer.
\begin{itemize}
 \item \textbf{NRP\textsubscript{direct} }. The encoding networks are $[500, 200, 100]$, $[1000, 500, 300, 100]$, $[500, 300, 100]$, and  $[500, 300, 100]$ in ml100k, ml1m, Amazon, and Ichiba datasets.\\
  The prediction network is $[500, 200, 100, 50 ,1]$ in all datasets.
 
 \item \textbf{DHA \cite{Li18}  and aSDAE \cite{Dong17}} . The encoders and decoders in these two methods are the same as NRP\textsubscript{direct}, which makes the comparisons fair. These two methods have a large number of hyper-parameters and we tried a large range of values to get the best results. We have set $\lambda_2=\lambda_3=0.5$ and $\lambda_1=1$. We use SGD with learning rate of $0.1$ for the optimization over the variables $\bU$ and $\bV$.
 \item  \textbf{NRP\textsubscript{DHA} } and  \textbf{NRP\textsubscript{aSDAE} }. The encoders and decoders in these two methods are the same as NRP\textsubscript{direct} , DHA, and aSDAE, which make the comparisons fair.
\item \textbf{HIRE \cite{Liu19}}. We used the code provided by the authors without changing the hyper-parameters.
 \item \textbf{NeuMF  \cite{He17}}. This method has one deep and one shallow branches. The structure of the deep branch is the same as the encoding networks of the NRP\textsubscript{direct}. The shallow branch, which uses the IDs as the input, computes the element-wise product of the representations, as explained in the original paper.
 \item \textbf{Autorec  \cite{Sedhain15}} . We implemented the I-Autorec, which reconstructs the ratings of the items. This method overfits to the training data fast, even using small networks. The autoencoder is $[m,500,m]$ in ml100k, as suggested by the original paper, and $[m,500,300,100,300,500,m]$ in ml1m, Amazon, and Ichiba datasets, where $m$ is the number of users.
\item \textbf{DSSM  \cite{Huang13} } Each branch has the same structure as the encoder network in our  NRP\textsubscript{direct}.
\item \textbf{ACCM} \cite{Shi18}. It uses an embedding layer to map user and item IDs to low-dimensional representations. To map side information to the low-dimensional representations, it uses the same structure as the encoder network of the NRP\textsubscript{direct}.

\item \textbf{ACCM\textsubscript{MLP} \cite{Shi18}.} Same as the ACCM above. To map the joint user and item representation to the rating, it uses the same structure as the prediction network of the NRP\textsubscript{direct}.

\end{itemize}

\begin{table}[!h]
\caption{Implementation details: optimization method (optimizer), learning rate (lr), regularization (regu.), activation function (activ.).}
\label{t:imp}
\begin{center}
\begin{tabular}[c]{ccccc@{\hspace{4ex}}cccc} 
\toprule
& \multicolumn{4}{c}{\dotfill ml100k \dotfill}  & \multicolumn{4}{c}{\dotfill ml1m \dotfill} \\
  & optimizer & lr  & regu.\ & activ.\  & optimizer & lr  & regu.\ & activ.\ \\
\midrule
NRP\textsubscript{direct}  & RMSprop & $0.001$ & $0$ & selu & SGD & $0.0005$ & $0$ & selu \\
\midrule
aSDAE &  SGD & $0.1$ & $10^{-5}$ & tanh &  RMSprop & $0.001$ & $10^{-5}$ & tanh  \\
\midrule
DHA &  SGD & $0.1$ & $10^{-5}$ & tanh   &  RMSprop & $0.001$ & $10^{-5}$ & tanh  \\
\midrule
NRP\textsubscript{aSDAE} &  RMSprop & $0.001$ & $0.03$ & tanh &  RMSprop & $0.001$ & $0$ & tanh  \\
\midrule
NRP\textsubscript{DHA}  &  RMSprop & $0.001$ & $0.03$ & tanh &  RMSprop & $0.001$ & $10^{-4}$ & tanh  \\
\midrule
Autorec & Adam & $0.001$ & $0.005$ & tanh & Adam & $0.001$ & $0.001$ & tanh  \\
\midrule
NeuMF & RMSprop & $0.001$ & $0.01$ & selu & RMSprop & $0.001$ & $0.01$ & selu \\
\midrule
DSSM & RMSprop & $0.001$ & $0.001$ & selu & SGD & $0.0005$ & $0$ & selu  \\
\midrule
MF & SGD & $0.1$ & $0.0001$ & NA & SGD & $0.1$ & $0.0001$ & NA  \\
\midrule
ACCM & RMSprop & $0.001$ & $0$ & tanh & SGD & $0.001$ & $0$ & tanh  \\
\bottomrule
\end{tabular}

\begin{tabular}[c]{ccccc@{\hspace{4ex}}cccc} 
\toprule
& \multicolumn{4}{c}{\dotfill Amazon \dotfill}  & \multicolumn{4}{c}{\dotfill Ichiba \dotfill} \\
  & optimizer & lr  & regu.\ & activ.\  & optimizer & lr  & regu.\ & activ.\ \\
\midrule
NRP\textsubscript{direct}  & RMSprop & $0.001$ & $0.001$ & selu & SGD & $0.0005$ & $0.001$ & selu \\
\midrule
aSDAE &  OM & OM & OM & OM &  OM & OM & OM & OM  \\
\midrule
DHA &  OM & OM & OM & OM &  OM & OM & OM & OM  \\
\midrule
NRP\textsubscript{aSDAE}  &  RMSprop & $0.001$ & $0.001$ & selu &  OM & OM & OM & OM  \\
\midrule
NRP\textsubscript{DHA}  &  RMSprop & $0.001$ & $0.001$ & tanh &   OM & OM & OM & OM  \\
\midrule
Autorec & RMSprop & $0.001$ & $0$ & selu & Adam & $0.001$ & $0.001$ & tanh  \\
\midrule
NeuMF & RMSprop & $0.00001$ & $0.001$ & selu & SGD & $0.0001$ & $0.01$ & selu \\
\midrule
DSSM&  NA & NA & NA & NA & SGD & $0.0005$ & $0.001$ & selu\\
\midrule
MF & SGD & $0.1$ & $0.0001$ & NA & SGD & $0.1$ & $0.0001$ & NA  \\
\bottomrule
\end{tabular}
\end{center}
\end{table}

\end{document}

%% file: abbr1.tex
\usepackage{bbm}
\newcommand{\ind}{\ensuremath{\mathbbm{1}}}

\newcommand{\bR}{\ensuremath{\mathbf{R}}}
\newcommand{\bbf}{\ensuremath{\mathbf{f}}}

\newcommand{\bg}{\ensuremath{\mathbf{g}}}

\newcommand{\bp}{\ensuremath{\mathbf{p}}}

\newcommand{\bX}{\ensuremath{\mathbf{X}}}
\newcommand{\bY}{\ensuremath{\mathbf{Y}}}
\newcommand{\bU}{\ensuremath{\mathbf{U}}}
\newcommand{\bV}{\ensuremath{\mathbf{V}}}
\newcommand{\bz}{\ensuremath{\mathbf{z}}}

\newcommand{\bE}{\ensuremath{\mathbf{E}}}
\newcommand{\bI}{\ensuremath{\mathbf{I}}}

\newcommand{\btheta}{\ensuremath{\boldsymbol{\theta}}}

\newcommand{\bbR}{\ensuremath{\mathbb{R}}}